\documentclass[a4paper,UKenglish,cleveref,autoref, thm-restate]{scrarticle}

\usepackage{amstext,amsfonts,amssymb,amsmath,amsthm,latexsym}
\usepackage[pdftex]{color,graphicx}

\usepackage{titlepic,subcaption}
\usepackage{float}
\usepackage{tikz}
\usetikzlibrary{calc,snakes,positioning,fit}

\usepackage{url}
\usepackage[colorlinks=true, linkcolor=black, citecolor=blue]{hyperref}
\usepackage[capitalize]{cleveref}
\usepackage[absolute]{textpos}

\newtheorem{lemma}{Lemma}[section]
\newtheorem{theorem}{Theorem}[section]

\newtheorem{claim}{Claim}[section]

\DeclareMathOperator{\tww}{\mathbf{tww}}
\DeclareMathOperator{\tw}{\mathbf{tw}}
\DeclareMathOperator{\bw}{\mathbf{bw}}

\title{Bounding twin-width for bounded-treewidth graphs, planar graphs, and bipartite graphs\thanks{%
This research is part of projects that has received funding from the European Research Council (ERC) under the European Union's Horizon 2020 research and innovation programme Grant Agreement 714704.
Initial part of the reseach was done when HJ was on an internship at University of Warsaw in Spring and Summer 2021.}}
\author{Hugo Jacob\footnote{ENS Paris-Saclay, France, \texttt{hugo.jacob@ens-paris-saclay.fr}.} \and
	Marcin Pilipczuk\footnote{Institute of Informatics, University of Warsaw, Poland, \texttt{malcin@mimuw.edu.pl}.}}
\date{}

\begin{document}
\maketitle

\begin{textblock}{20}(0.9, 11.6)
\includegraphics[width=30px]{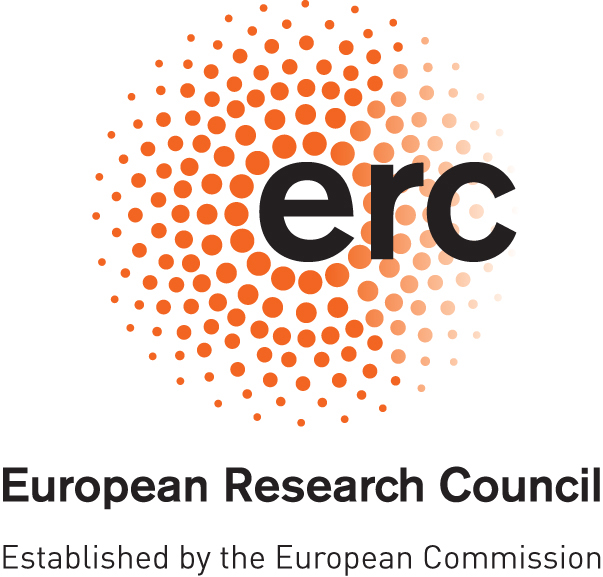}%
\end{textblock}
\begin{textblock}{20}(0.9, 12.3)
\includegraphics[width=30px]{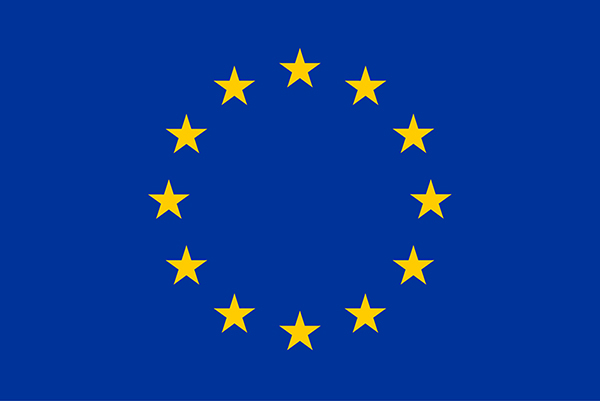}%
\end{textblock}

\begin{abstract}
Twin-width is a newly introduced graph width parameter that aims at generalizing a 
wide range of ``nicely structured'' graph classes. 
In this work, we focus on obtaining good bounds on twin-width $\tww(G)$ for graphs $G$ from a number of classic graph classes.
We prove the following:
\begin{itemize}
	\item $\tww(G) \leq 3\cdot 2^{\tw(G)-1}$, where $\tw(G)$ is the treewidth of $G$,
	\item $\tww(G) \leq \max(4\bw(G),\frac{9}{2}\bw(G)-3)$ for a planar graph $G$ with $\bw(G) \geq 2$, where $\bw(G)$ is the branchwidth of $G$,
	\item $\tww(G) \leq 183$ for a planar graph $G$,
	\item the twin-width of a universal bipartite graph $(X,2^X,E)$ with $|X|=n$ is $n - \log_2(n) +
	\mathcal{O}(1)$ .
\end{itemize}

An important idea behind the bounds for planar graphs is to use an embedding of the graph and sphere-cut decompositions to obtain good bounds on neighbourhood complexity. 
\end{abstract}

\section{Introduction}

Twin-width is a graph parameter recently introduced by Bonnet et al \cite{Twinwidth1}, which has
already proven to be very versatile and useful. It is defined via iterated contraction of vertices that are
almost twins, while limiting the amount of errors that are carried on. Twin-width is known, for
instance, to be bounded on classes of graphs of bounded treewidth, bounded rank-width, or excluding
a fixed minor \cite{Twinwidth1}. It is also possible to design algorithms on the contraction
sequences, thus providing a common framework for efficient algorithms on several graph classes
\cite{Twinwidth1,Twinwidth3}. Twin-width is also linked to First Order logic, FO model checking is
FPT for graphs of bounded twin-width, and FO transductions preserve twin-width boundedness
\cite{Twinwidth1} (see also \cite{GPT21}).
However, finding good contraction sequences is hard \cite{TwwNPc}, and the arguments used to show
the boundedness of twin-width do not necessarily provide a constructive way of obtaining a
contraction sequence. This motivates more detailed comparisons of twin-width to other parameters
(the case of poset width has already been considered \cite{BalabanHlineny21} for instance). 

Many currently known bounds on the twin-width, in particular for minor-closed graph classes such as planar graphs,
     rely on very general arguments and result in unreasonably large constants. 
Finding a better bound was explicitly mentionned as an open problem.
In this paper, we present a few results we obtained while looking for an improved bound.

We first give some results on graphs of bounded treewidth: an exponential bound on the
twin-width of a graph of bounded treewidth, and a linear bound on the twin-width of planar graphs of
bounded treewidth. We then obtain a bound of $183$ on the twin-width of planar graphs, which is, to
the best of our knowledge, currently the best known bound.
We were not able to prove a matching exponential lower bound for the twin-width of graphs of bounded treewidth.
As a partial result in this direction, we determine the twin-width of
universal bipartite graphs up to a constant additive term.

Independently of this work, Bonnet, Kwon, and Wood~\cite{kwon}
obtained a bound of $583$ on the twin-width of planar graphs, among other
results on more general classes such as bounded genus graphs.

\section{Preliminaries}

In the following $[n]$ denotes $\{1,\ldots,n\}$. Given a set $X$, $|X|$ denotes its cardinality and
$2^X$ denotes the set of subsets of $X$.

The subgraph induced by vertex subset $A$ in graph $G$ is denoted by $G[A]$, $G - A$ denotes $G[V
\setminus A]$. The neighbourhood of vertex $v$ in $G=(V,E)$ is $N(v)=\{w \in V | \{v,w\} \in E\}$,
and we extend this notation with $N(X)=\bigcup_{x \in X} N(x)$. To emphasize that the neighbourhood is
taken in graph $G$, we use $N_G$ instead of $N$. 

We call neighbourhood classes with respect to $Y$ in $X$ the set $\Omega(X,Y)=\{N(x) \cap Y : x \in
X\}$. Note that if $|Y|=k$, then $|\Omega(X,Y)|\leq 2^k$.

We call universal bipartite graph the bipartite graph $\mathcal{B}(n)=([n],2^{[n]},\{(k,A \cup
\{k\}) : k \in [n], A \in 2^{[n] \setminus \{k\}}\})$.

We now define formally the notion of \emph{twin-width} of a graph. A \emph{trigraph} is a triple
$G=(V,E,R)$ where $E$ and $R$ are disjoint sets of edges on $V$, the (usual) edges and the red edges
respectively. The notion of induced subgraph is extended to trigraphs in the obvious way.
We denote by $R(v)$ the red neighbourhood of $v$. A trigraph $(V,E,R)$
such that $(V,R)$ has maximum degree at most $d$ is a $d$-trigraph. Any graph $(V,E)$ can be seen as
the trigraph $(V,E,\varnothing)$.
Given a trigraph $G=(V,E,R)$ and two vertices $u,v$ of $V$, the trigraph $G'=(V',E',R')$ obtained by
the contraction\footnote{The vertices are not required to be adjacent.} of $u,v$ into a new vertex $w$
is defined as the trigraph on vertex set $V'=V\setminus\{u,v\} \cup \{w\}$, such that $G -
\{u,v\}=G'-\{w\}$, and such that $N_{G'}(w)=N_G(u) \cap N_G(v)$ and $R_{G'}(w)=R_G(u) \cup R_G(v)
\cup (N_G(u) \Delta N_G(v))$, where $\Delta$ denotes the symmetric difference.
A \emph{$d$-contraction sequence} of $G$ is a sequence of trigraph contractions starting with $G$
ending with the single-vertex trigraph, such that all intermediate trigraphs have maximum red degree
$d$. The \emph{twin-width} of graph $G$ is the minimum $d$ such that there exists a $d$-contraction
sequence, it is denoted $\tww(G)$.

We use the notation of \cite{Courcelle18a} for tree decompositions.
Given a rooted tree $T$, $N_T$ denotes its nodes, $\leq_T$ denotes its \emph{ancestor relation}
which is a partial order on $N_T$ where the root is the maximal element, and the leaves are the
minimal elements.
For a fixed node $u$ of $T$, we denote by $p(u)$ its parent (minimal strict ancestor), by
$T_\leq(u)$ the set $\{w \in N_T | w \leq_T u\}$ and similarly for $T_<(u), T_\geq(u), T_>(u)$.
A tree $T$ is \emph{normal} for graph $G$ if $V(G)=N_T$, and for each edge of $G$, its endpoints are
comparable under $<_T$.
We denote by $(T,f)$ a \emph{tree decomposition} of $G$ where $T$ is a rooted tree, $f$ maps $N_T$ to
$2^{V(G)}$ and satisfies the following conditions: every vertex of $G$ is contained in at least one
\emph{bag} $f(u)$, for every edge of $G$ there is a bag containing its two endpoints, and for every
vertex of $G$, the nodes $u$ such that $f(u)$ contains it induced a connected subgraph of $T$. 
$(T,f)$ is \emph{normal} if $T$ is normal for $G$, $f(u) \subseteq T_\geq(u)$ and $u \in f(u)$, for
every $u \in N_T$. $f^*(u)$ denotes $f(u) \setminus \{u\}$. 
$(T,f)$ is \emph{clean} if it is normal, $f^*(u)=N_G(T_\leq(u))\cap T_>(u)$ for every node $u$ of
$T$, and $p(u) \in f(u)$ for every node $u$ of $T$ except its root.
The width of $(T,f)$ is $\max_{u \in N_T} |f(u)|-1$, and the treewidth of a graph is the minimum
width over its tree decompositions. It is denoted by $\tw(G)$.

Let $\Sigma$ be a sphere $\{(x,y,z) \in \mathbb{R}^3 | x^2+y^2+z^2=1\}$. A \emph{$\Sigma$-plane} graph $G$
is a planar graph embedded in $\Sigma$ without crossing edges. To simplify notations, we do not
distinguish vertices and edges from the points of $\Sigma$ representing them. An $O$-arc is a subset of $\Sigma$
homeomorphic to a circle. An $O$-arc in $\Sigma$ is a \emph{noose} if it meets $G$ only in
vertices and intersects every face at most once. The length of a noose is the number of vertices it
meets. Every noose $O$ bounds two open discs $\Delta_1,\Delta_2$ in $\Sigma$, i.e., $\Delta_1 \cap
\Delta_2 = \varnothing$ and $\Delta_1 \cup \Delta_2 \cup O = \Sigma$.

A \emph{branch decomposition} $(T,\mu)$ of a graph $G$ consists of a ternary tree $T$ (internal
vertices of degree $3$) and a bijection $\mu : L \to E(G)$ from the set $L$ of leaves of $T$ to the
edge set of $G$. For every edge $e$ of $T$, the \emph{middle set} of $e$ is a subset of $V(G)$
corresponding to the common vertices of the two graphs induced by the edges associated to the leaves
of the two connected components of $T - e$. The width of the decomposition is the maximum
cardinality of the middle sets over all edges of $T$. An optimal decomposition is one with minimum
width, which is called \emph{branchwidth} and denoted by $\bw(G)$.

For a $\Sigma$-plane graph $G$, a \emph{sphere-cut decomposition} $(T,\mu)$ is a branch
decomposition such that for every edge $e$ of $T$, there exists a noose $O_e$ meeting $G$ only on
the vertices of the middle set of $e$ and such that the two graphs induced by the edges associated
to the leaves of the two connected components of $T - e$ are each on one side of $O_e$. 
The following result is stated in \cite{ExploitingSphereCut} as a consequence of the results of 
Seymour and Thomas \cite{Ratcatcher}, and Gu and Tamaki \cite{GuTamaki08}.

\begin{lemma}
	Let $G$ be a connected $\Sigma$-plane graph of branchwidth at most $\ell$ without vertices of
	degree one. There exists a sphere-cut decomposition of $G$ of width at most $\ell$, and it can
	be computed in time $\mathcal{O}(|V(G)|^3)$.
\end{lemma}

A sphere-cut decomposition $(T,\mu)$ can be rooted by subdividing an edge $e$ of $T$ into two edges
$e',e''$ with middle vertex $s$, and adding a root $r$ connected to $s$. The middle set of $e'$ and
$e''$ is the middle set of $e$, and $\{r,s\}$ has an empty middle set. For every edge $e$ of $T$,
the subtree that does not contain the root is called the \emph{lower part}, we denote by $G_e$ the
subgraph induced by the edges associated to the leaves of the lower part. For an internal node $v$
of $T$, the edge on the path to $r$, is called the parent edge, and the other two are called
children edges. There can be at most $2$ vertices common to the middle sets of these three edges
\cite{ExploitingSphereCut}. 

We slightly extend sphere-cut decompositions to cover the case of connected graphs with minimal
degree one and branchwidth at least 2. Consider a connected graph $G$, let $G'$ be its maximal
induced subgraph with no vertex of degree one. Note that $G'$ must be connected and that the graph
$H$ induced by the edges $E(G) \setminus E(G')$ is a forest where each tree has only a vertex in
common with $G'$, which we will consider as its root. We can first compute a sphere-cut
decomposition $(T',\mu')$ of $G'$ and then for each root $r$ of a tree $H_i$ in $H$, we can find an
edge $e$ of $T'$ such that $r$ is in its middle set (it exists because $r$ has degree at least $2$ in
$G'$), and attach an optimal branch decomposition of $H_i$ on $e$. This does not increase the branchwidth
because $r$ was already in the middle set of $e$. Once this is done for all trees $H_i$ in $H$, we
obtain a branch decomposition $(T,\mu)$ of $G$, such that there exists a noose meeting exactly the
middle set of each edge of $T$. However, the nooses do not correspond to cycles in the radial graph
anymore since we have to embed the $H_i$ in faces of $G'$.

\begin{lemma}
	Let $G$ be a connected $\Sigma$-plane graph of branchwidth $\ell \geq 2$. There exists a
	sphere-cut decomposition of $G$ of width $\ell$, and it can be computed in time
	$\mathcal{O}(|V(G)|^3)$.
\end{lemma}

\begin{proof}

Computing $G'$ and $H$ can be done in time $\mathcal{O}(|E(G)|)$ and the optimal decompositions of
the trees in $H$ can be produced in total time $\mathcal{O}(|V(G)|)$.

\end{proof}

\section{Twin-width of graphs of bounded treewidth}

The following result reuses a method to bound clique-width described in \cite[Proposition
13]{Courcelle18a}.

\begin{lemma}\label{lemma:upperbound-tww-tw}
	For an undirected graph $G$, $\tww(G) \leq 3\cdot 2^{\tw(G)-1}$.
\end{lemma}

\begin{proof}
	We consider a connected graph $G$ as the twin-width of a disconnected graph is simply the
	maximum twin-width of one of its connected components.
	
	We consider a clean tree decomposition $(T,f)$ of $G$ of width $\tw(G)$ (this
	is always possible \cite[Lemma 3, Lemma 5]{Courcelle18a}).
	
	We proceed by induction on the tree $T$.	
	Consider node $v$ with children $u_1,\dots,u_k$

	We assume that for each $u_i$, we have contracted $V(T_{\leq}(u_i))$ into $A_i$ consisting of at
	most $|\Omega(T_{\leq}(u_i),f^*(u_i))|$ vertices such that their incident red edges have both
	endpoints within $A_i$.

	We will contract these sets of vertices into a set $C$ consisting of at most
	$|\Omega(T_{\leq}(v),f^*(v))|$ vertices.

	Let $B_0=\varnothing$.
	We will inductively obtain for each $i \in [k]$ a vertex set $B_i$ of size at most
	$\left|\Omega\left(\bigcup \limits_{j=1}^i T_{\leq}(u_j),f^*(v)\right)\right|$, by contracting
	vertices of $\bigcup \limits_{j=1}^i A_i$.

	For each $i \in [k]$, we first contract vertices of $A_i$ that have the same neighbourhood in
	$f^*(v)$, this produces $\widetilde{A}_i$ consisting of at most
	$|\Omega(T_{\leq}(u_i),f^*(u_i)-\{v\})|$ vertices.  Doing so will produce at most
	$|\widetilde{A}_i|$ red edges incident to $v$, which now has at
	most $|B_{i-1}|+|\widetilde{A}_i|$ incident red edges.
	We then contract vertices of $\widetilde{A}_i \cup B_{i-1}$ that have the same neighbourhood in
	$f^*(v)$, producing $B_i$ consisting of at most $\left|\Omega\left(\bigcup \limits_{j=1}^i
	T_{\leq}(u_j),f^*(v)\right)\right|$ vertices. Note that the red degree of a vertex resulting
	from one of these contractions is at most $|\widetilde{A}_i|-1 + |B_{i-1}|-1 + |\{v\}| \leq
	|B_{i-1}| + |\widetilde{A}_i|$.
	Vertex $v$ now has $|B_i|$ incident red edges. 
	
	After this we can contract $v$ with the vertex of $B_k$ having the same neighbourhood in
	$f^*(v)$ if it exists. This produces $C$ consisting of at most $|\Omega(T_{\leq}(v),f^*(v))|$
	vertices and such that their incident red edges remain within $C$.

	In all of the described steps, the red degree of a vertex is at most $3 \cdot 2^{\tw(G)-1}$:
	\begin{itemize}
		\item Vertices in $A_i$ have red degree at most $|A_i| \leq |\Omega(T_{\leq}(u_i),f^*(u_i))|
		\leq 2^{\tw(G)}$.

		\item Vertices in $\widetilde{A}_i$ have red degree at most $|\widetilde{A}_i| \leq
		|\Omega(T_{\leq}(u_i),f^*(u_i)-\{v\})| \leq 2^{\tw(G)-1}$.

		\item $v$ has red degree at most $|B_{i-1}|+|\widetilde{A}_i| \leq \left|\Omega\left(\bigcup
		\limits_{j=1}^{i-1} T_{\leq}(u_j),f^*(v)\right)\right| +
		|\Omega(T_{\leq}(u_i),f^*(u_i)-\{v\})| \leq 3 \cdot 2^{\tw(G)-1}$.

		\item When contracting $B_{i-1} \cup \widetilde{A}_i$, vertices have red degree at most
		$|B_{i-1}| + |\widetilde{A}_i| \leq 3 \cdot 2^{\tw(G)-1}$
	\end{itemize}

	Since the property is trivial on leaves of the tree, we conclude that $\tww(G) \leq 3\cdot
	2^{\tw(G)-1}$.
\end{proof}

Using sphere-cut decompositions, we establish the following lemma.

\begin{lemma}\label{lemma:tww-bw-planar}
	For an undirected connected planar graph $G$ with $\bw(G) \geq 2$, $\tww(G) \leq
	\max(4\bw(G),\frac{9}{2}\bw(G)-3) \leq \max(4\tw(G) + 4,\frac{9}{2}\tw(G)+\frac{3}{2})$.

	For an undirected connected planar graph $G$ with $\bw(G) \leq 1$, $\tww(G)=0$.
\end{lemma}

This mainly relies on the following result.

\begin{claim}\label{claim:neighbourhood-outerplanar}
	If $N$ is a noose with $|V(N)|=k$ that separates a plane graph $G$ in $G_1$ and $G_2$, then
	$\Omega(V(G_1) \setminus V(G_2) ,V(G_2)) = \Omega(V(G_1) \setminus V(N),V(N))$ and
	$|\Omega(V(G_1)\setminus V(N),V(N))| \leq 4k-4 =:h(k)$.
\end{claim}

\begin{proof}
	We will count the different possible neighbourhoods by size:
	\begin{itemize}
		\item The only possibility for size 0 is the empty neighbourhood.

		\item The possibilities for size 1 are the singletons of $V(N)$ and there are $k$ of them.

		\item For the neighbourhoods of size 2, we pick one vertex for each of them, and call $A$
		the set of picked vertices. We now consider $G_1[A\cup V(N)]$ and smooth the vertices of $A$
		in it, i.e. for each vertex $a$ of $A$ with incident edges $ua,av$, we remove vertex $a$
		and edges $ua,av$ and replace them by edge $uv$, this operation preserves planarity and the
		resulting graph $H$ is an outerplanar graph on vertices $V(N)$ because they were on the
		outerface of $G_1[A \cup V(N)]$. Since the number of edges of $H$ is at most $2k-3$ because
		it is outerplanar and is equal to $|A|$, the number of different neighbourhoods is bounded
		by $2k-3$.

		\item For the neighbourhoods of size more than 3, we once again pick one vertex for each of
		them, and call $B$ the set of picked vertices. We now consider $G_1[B \cup V(N)]$ which is
		planar. We show $|B|\leq n_3(k)\leq k-2$ by induction on $k=V(N)$, where $n_3(k)$ denotes
		the maximum number of vertices of $B$ of degree more than 3 we can have in $G_3[B\cup
		V(N)]$. First, if $k \leq 2$ then there are no such neighbourhoods, and if $k=3$, there is
		exactly one. Then for $k>3$, $$n_3(k)=1+\max\left\{ \sum \limits_{i=1}^\ell n_3(a_i+1) : \ell \geq
		3, \forall i \in [\ell], a_i \geq 1, \sum \limits_{i=1}^\ell a_i = k \right\}$$
		because after placing one vertex $v$ of degree $\ell \geq 3$, we must have subdivided our
		instance into $\ell$ smaller instances because edges incident to v will not be crossed by
		other edges. Using the induction hypothesis, we have
		$$n_3(k) \leq 1 + \sum \limits_{i=1}^\ell (a_i - 1) \leq 1 + k - l \leq k - 2$$
	\end{itemize}
	By summing the previous bounds, we conclude that $|\Omega(V(G_1) \setminus V(N),V(N))|\leq 4k - 4$
\end{proof}

Note that this bound is tight: denote the vertices in their order on the noose by $[k]$, we can
place vertices with neighbourhoods $\{\varnothing\} \cup \{ \{i\} : i \in [k] \} \cup \{ \{i,i+1\} :
i \in [k-1]\} \cup \{ \{1,i,i+1\}, \{1,i+1\} : i \in [2,k-1] \}$.

\begin{proof}[Proof of \cref{lemma:tww-bw-planar}]
	Consider a connected planar graph $G$.
	If $G$ has branchwidth at most $1$, it cannot contain a path on $4$ vertices as a subgraph,
	hence it is a star and has twin-width $0$ (first contract twins and finish with the root).
	We now consider the case when $\bw(G) \geq 2$.
	$G$ admits a sphere-cut decomposition $(T,\mu)$ of width $k:=\bw(G)$.
	We root $T$ arbitrarily.
	We proceed by induction on $T$.
	Consider a parent edge $e$ with children edges $e_1,e_2$.
	We assume that, for $i \in \{1,2\}$, $V(G_{e_i}-V(N_{e_i}))$, has been contracted to a set $A_i$
	according to the neighbourhood in $V(N_{e_i})$. Consequently, $|A_i|$ is at most
	$|\Omega(V(G_{e_i}-V(N_{e_i})),V(N_{e_i}))|$, and red edges incident to $A_i$ have both
	endpoints in $A_i$.

	Let $x:=|V(N_e) \cap V(N_{e_1})|$ and $y:=|V(N_e)\cap V(N_{e_2})|$. Note that $x+y-2 \leq
	|V(N_e)| \leq k$.
	
	Let $I:=V(N_{e_1}) \cap V(N_{e_2}) \setminus V(N_e)$, and $z:=|I|$
	
	For $i \in \{1,2\}$, we contract vertices of $A_i$ that have the same neighbourhood in
	$V(N_{e_i}) \setminus I$, and call the resulting set of vertices $\widetilde{A}_i$.
	The vertices of $I$ now have red degree at most $|\widetilde{A}_1| + |\widetilde{A}_2|$, while the vertices of $\widetilde{A}_i$
	have red degree at most $|I| + |\widetilde{A}_i|-1$.

	We then contract the vertices of $I \cup \widetilde{A}_1 \cup \widetilde{A}_2$ that have the
	same neighbourhood in $V(N_e)$, and call $A$ the resulting set of vertices.
	Contracted vertices have red degree at most $|\widetilde{A}_1| + |\widetilde{A}_2| + |I|-2$.
	Using Claim \ref{claim:neighbourhood-outerplanar}, we obtain the following inequalities:

	$|\widetilde{A}_1| + |\widetilde{A}_2| \leq |\Omega(V(G_{e_1}-V(N_{e_1})),V(N_{e_1})\setminus I)| +
	|\Omega(V(G_{e_2}-V(N_{e_2})),V(N_{e_2})\setminus I)|\leq (4x-4) + (4y-4) \leq 4k$

	$|\widetilde{A}_1| + |\widetilde{A}_2| + |I| - 2 \leq
	|\Omega(V(G_{e_1}-V(N_{e_1})),V(N_{e_1})\setminus I)| +
	|\Omega(V(G_{e_2}-V(N_{e_2})),V(N_{e_2})\setminus I)| + z - 2 \leq 4x + 4y + z - 10 =
	\frac{7}{2}(x+y) + \frac{1}{2}(x+z) + \frac{1}{2}(y+z) - 10$

	We have the following constraints on $x,y,z$: $x+y \leq k+2$, $|V(N_{e_1})| = x+z \leq k$,
	$|V(N_{e_2})| = y+z \leq k$.

	By summing inequalities, we obtain $|\widetilde{A}_1| + |\widetilde{A}_2| + |I| - 2 \leq \frac{9}{2}k - 3$.

	$V(G_e-V(N_e))$ has been contracted to a set $A$ of at most $|\Omega(V(G_e-V(N_e)),V(N_e))|$
	vertices.

	We conclude that $tww(G) \leq \max(4k,\frac{9}{2}k-3)$
\end{proof}

\section{Twin-width of planar graphs}

\begin{theorem}
	The twin-width of planar graphs is at most $183$.
\end{theorem}

\begin{proof}
	We will make use of the argument used to decompose planar graphs in \cite[Lemma
	5]{PlanarStrongProduct21}, and produce a $d$-contraction sequence of a planar graph $G$
	inductively on the decomposition, $d \leq 183$. The embedding of the graph will be useful in our
	arguments to make use of Claim \ref{claim:neighbourhood-outerplanar}. Recall that $h(k)=4k-4$.

	We may suppose that $G$ is connected since the twin-width of a graph is simply the maximum of
	the twin-width over its connected components. We denote by $G^+$ a triangulation containing $G$
	as a spanning subgraph. Let $T$ be a BFS spanning tree in $G^+$ with root $r$ on its outerface.
	Note that since $G$ is a subgraph of $G^+$, the plane embedding of $G^+$ gives a plane embedding
	of $G$ and its subgraphs.

	For a cycle $C$, we write $C=[P_1,\dots,P_k]$ if the $P_i$ are pairwise disjoint, and the last
	vertex of $P_i$ is adjacent to the first vertex of $P_{i+1}$ for $i \in [k]$, with $P_{k+1}=P_1$.
	For a path $P$, we write $P=[P_1,\dots,P_k]$ if the $P_i$ are pairwise disjoint, and the last
	vertex of $P_i$ is adjacent to the first vertex of $P_{i+1}$ for $i \in [k-1]$.

	The following version of Sperner's Lemma is used to recursively decompose $G^+$.
	\begin{lemma}[Sperner's Lemma]
		Let $G$ be a near-triangulation whose vertices are coloured $1,2,3$, with the outerface
		$F=[P_1,P_2,P_3]$ where each vertex in $P_i$ is coloured $i$. Then $G$ contains an internal
		face whose vertices are coloured $1,2,3$.
	\end{lemma}

	We prove inductively the following:

	\begin{lemma}
		Let $P_1,\ldots,P_k$ for some $k \in [5]$ be pairwise disjoint vertical paths of $T$ such
		that $F=[P_1,\ldots,P_k]$ is a cycle in $G^+$, let $H$ be the subgraph of $G$ induced by
		the vertices of $F$ and the set $X$ of vertices in the (strict) interior of $F$, with $r
		\notin X$. Let $X^j$ denote the set of vertices of $X$ that are at a distance $j$ from $r$
		in $T$. We can construct a partial $d$-contraction sequence of $H$ to $H'$ such that for
		each $j$, the vertices of $X^j$ are contracted to obtain a set of vertices $A^j$ in $H'$,
		$|A^j| \leq h(3k)$, the vertices of $A^j$ have red neighbours only in $A^{j-1},A^j,A^{j+1}$,
		and $d \leq 183$. 
	\end{lemma}
	
	\begin{proof}
		If we have 3 vertices then there is no vertex in the interior of the triangle, the empty
		contraction sequence satisfies the properties.

		Otherwise, we decompose $H$ using the argument of \cite{PlanarStrongProduct21}, see
		\cref{fig:decomposition}. First, we
		coulour the vertices of $H$ with $k$ coulours as follows. For each vertex $v \in V(H)$, we
		assign coulour $i \in [k]$ if the first vertex of $F$ on the path from $v$ to $r$ in $T$ is a
		vertex of $P_i$. This is well defined because $r$ is on the outerface of $G^+$.
		
		We set up for Sperner's Lemma with the following constructions:
		\begin{itemize}
			\item If $k=1$ then, since $F$ is a cycle, $P_1$ has at least 3 vertices so we can
			write $P_1=[u,R_2,v]$, and set $R_1:=u,R_3:=v$.
			\item If $k=2$ then, since $F$ is a cycle, one of $P_1$ and $P_2$ has at least 2
			vertices. W.l.o.g. assume it is $P_1$, then we write $P_1=[u,R_2]$, and set
			$R_1:=u,R_3:=P_2$.
			\item If $k=3$ then set $R_1:=P_1,R_2:=P_2,R_3:=P_3$.
			\item If $k=4$ then set $R_1:=P_1,R_2:=P_2,R_3:=[P_3,P_4]$
			\item If $k=5$ then set $R_1:=P_1,R_2:=[P_2,P_3],R_3:=[P_4,P_5]$
		\end{itemize}
		Note that $F=[R_1,R_2,R_3]$. We give colour $i$ to the vertices of $H$ whose first vertex of
		$F$ on their path to the root in $T$ is in $R_i$.

		Applying Sperner's Lemma, we obtain a triangular face of $G^+$, with vertices $v_1, v_2,
		v_3$ of the 3 coulours. We denote $Q'_i$ the path in $T$ from $v_i$ to $r$ restricted to its
		vertices in $X$ (it might be empty). These paths delimit at most 3 faces $F_1,F_2,F_3$, each
		of which having at most 5 vertical paths around it.
		We can apply the induction hypothesis on each of the faces to obtain partial contraction
		sequences. We first apply them in an arbitrary order (the contents of the faces are antiadjacent to each other). We denote by $A_\alpha^j$ the contracted sets of face $F_\alpha$.

		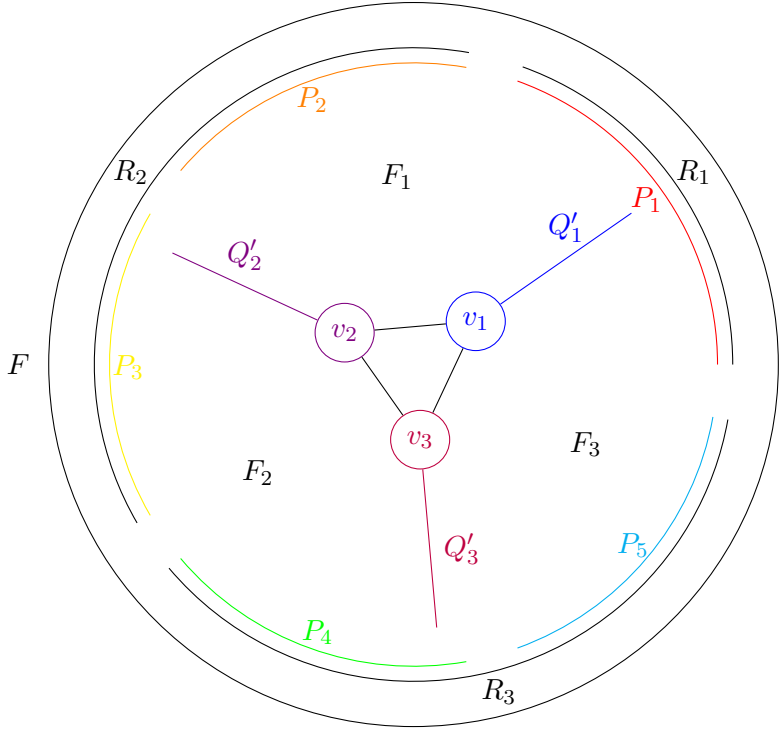
\begin{figure}
			\begin{tikzpicture}

			\node (11) at (0:4) {};
			\node (12) at (70:4) {};
			\node (21) at (80:4) {};
			\node (22) at (140:4) {};
			\node (31) at (150:4) {};
			\node (32) at (210:4) {};
			\node (41) at (220:4) {};
			\node (42) at (280:4) {};
			\node (51) at (290:4) {};
			\node (52) at (350:4) {};
			\node[draw,circle,color=blue] (v1) at (35:1) {$v_1$};
			\node[draw,circle,color=violet] (v2) at (155:1) {$v_2$};
			\node[draw,circle,color=purple] (v3) at (275:1) {$v_3$};
			\node (F1) at (95:2.5) {$F_1$};
			\node (F2) at (215:2.5) {$F_2$};
			\node (F3) at (335:2.5) {$F_3$};
			\node (R11) at (0:4.2) {};
			\node (R12) at (70:4.2) {};
			\node (R21) at (80:4.2) {};
			\node (R22) at (210:4.2) {};
			\node (R31) at (220:4.2) {};
			\node (R32) at (350:4.2) {};

			\draw (v1) -- (v2) -- (v3) -- (v1);

			\node (F) at (180:5.2) {$F$};
			\draw circle[radius=4.8];
			\draw[color=blue] (v1) -- (35:3.5) node[midway,yshift=12] {$Q'_1$};
			\draw[color=violet] (v2) -- (155:3.5) node[midway,yshift=12] {$Q'_2$};
			\draw[color=purple] (v3) -- (275:3.5) node[midway,xshift=12] {$Q'_3$};

			\draw[color=red] (11) arc[start angle=0, end angle=70, radius=4] (12)
			node[midway,xshift=-6,yshift=-3] {$P_1$};
			\draw[color=orange] (21) arc[start angle=80, end angle=140, radius=4] (22)
			node[midway,xshift=1,yshift=-7] {$P_2$};
			\draw[color=yellow] (31) arc[start angle=150, end angle=210, radius=4] (32)
			node[midway,xshift=7,yshift=-1] {$P_3$};
			\draw[color=green] (41) arc[start angle=220, end angle=280, radius=4] (42)
			node[midway,xshift=3,yshift=6] {$P_4$};
			\draw[color=cyan] (51) arc[start angle=290, end angle=350, radius=4] (52)
			node[midway,xshift=-5,yshift=5] {$P_5$};
			\draw (R11) arc[start angle=0, end angle=70, radius=4.2] (R12)
			node[midway,xshift=7,yshift=4] {$R_1$};
			\draw (R21) arc[start angle=80, end angle=210, radius=4.2] (R22)
			node[midway,xshift=-8,yshift=4] {$R_2$};
			\draw (R31) arc[start angle=220, end angle=350, radius=4.2] (R32)
			node[midway,xshift=1,yshift=-8] {$R_3$};

			\end{tikzpicture}
			\centering
			\caption{Decomposition of $F$}
			\label{fig:decomposition}
		\end{figure}

		For each $\alpha$ and increasing $j$, we contract all vertices of $A_\alpha^j$ that are in
		the same neighbourhood class with respect to $P_1,\ldots,P_k$ in $G$. Note that only vertices on
		layers $j-1,j,j+1$ of the $P_i$ may be adjacent and that there are at most 3 of the $P_i$s
		that are adjacent to $F_\alpha$. This gives us sets $\widetilde{A}_\alpha^j$ of size at most
		$h(9)$ by Claim \ref{claim:neighbourhood-outerplanar}, since by removing the vertices of
		$Q'_i$ and keeping only vertices of layers $j-1,j,j+1$ we obtain a graph that is still
		planar and in which the cycle delimiting $F_\alpha$ gives a noose with at most $9$ vertices
		(vertical paths have at most 1 vertex per layer).

		Then for increasing $j$, we contract vertices of $\widetilde{A}_1^j \cup \widetilde{A}_2^j \cup \widetilde{A}_3^j \cup Q'_1 \cup
		Q'_2 \cup Q'_3$ that are in the same neighbourhood class with respect to $P_1,\ldots,P_k$ in
		$G$, see \cref{fig:contraction}. This gives sets $A^j$ of size at most $h(15)$ by Claim
		\ref{claim:neighbourhood-outerplanar}, because we can deduce a noose from
		$F=[P_1,\ldots,P_k]$ and by keeping only the vertices of layers $j-1,j,j+1$ we have at most
		$15$ vertices on the noose.

		\begin{figure}
			\begin{tikzpicture}
			\tikzset{
			group/.style = {circle,draw,minimum size=#1, inner sep=0pt, outer sep=0pt},
			node/.style = {circle,fill,minimum size=5pt, inner sep=0pt, outer sep=0pt}
			  }	
			\node[group=40pt] (0) at (0,0) {$A^{j-1}$};
			
			\node[node,label=above left:$q_1^j$] (q1) at (3,2.5) {};
			\node[group=30pt] (A1) at (3,1.5) {$\widetilde{A}_1^j$};
			\node[node,label=above left:$q_2^j$] (q2) at (3,0.5) {};
			\node[group=30pt] (A2) at (3,-0.5) {$\widetilde{A}_2^j$};
			\node[node,label=left:$q_3^j$] (q3) at (3,-1.5) {};
			\node[group=30pt] (A3) at (3,-2.5) {$\widetilde{A}_3^j$};
			\node (a) at (3,3) {};\node (b) at (4.2,2.9) {};
			\node (c) at (3,-3.4) {};\node (d) at (4.2,-3.3) {};

			\node[node,label=above right:$q_1^{j+1}$] (q'1) at (6,2.5) {};
			\node[group=30pt] (A'1) at (6,1.5) {$\widetilde{A}_1^{j+1}$};
			\node[node,label=above right:$q_2^{j+1}$] (q'2) at (6,0.5) {};
			\node[group=30pt] (A'2) at (6,-0.5) {$\widetilde{A}_2^{j+1}$};
			\node[node,label=right:$q_3^{j+1}$] (q'3) at (6,-1.5) {};
			\node[group=30pt] (A'3) at (6,-2.5) {$\widetilde{A}_3^{j+1}$};
			\node (a') at (6,3) {};\node (b') at (4.8,2.9) {};
			\node (c') at (6,-3.4) {};\node (d') at (4.8,-3.3) {};

			\node[draw,dashed,blue,fit={(q1) (A1) (q2) (A2) (q3) (A3)}] (box) {};
			\node[blue] (Aj) at (1.7,-2.6) {$A^j$};

			\draw (q1) edge (0) edge (A1) edge (q'1) edge (A'1) edge (a) edge (b);
			\draw (A1) edge (0) edge (q2) edge (A'1) edge (q'2);
			\draw (q2) edge (0) edge (A2) edge (q'2) edge (A'2);
			\draw (A2) edge (0) edge (q3) edge (A'2) edge (q'3);
			\draw (q3) edge (0) edge (A3) edge (q'3) edge (A'3);
			\draw (A3) edge (0) edge (A'3) edge (c) edge (d);

			\draw (q'1) edge (A1) edge (A'1) edge (a') edge (b');
			\draw (A'1) edge (q2) edge (q'2);
			\draw (q'2) edge (A2) edge (A'2);
			\draw (A'2) edge (q3) edge (q'3);
			\draw (q'3) edge (A3) edge (A'3);
			\draw (A'3) edge (c') edge (d');

			\end{tikzpicture}
			\centering
			\caption{Second phase of the contraction procedure}
			\label{fig:contraction}
		\end{figure}
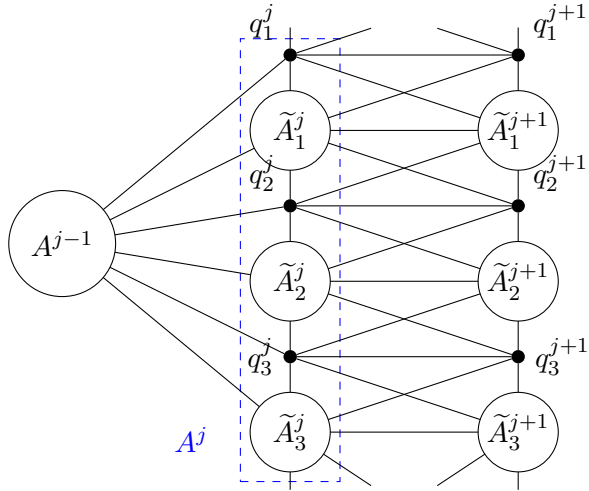

		We now bound the red degree that may appear in our contraction sequence.
		When first contracting $A_\alpha^j$ the number of red edges of its vertices is at most
		$|\widetilde{A}_\alpha^{j-1}|+|A_\alpha^j|+|A_\alpha^{j+1}|+6-2$ where the $6$ term bounds
		the number of vertices on the $Q'_i$ that are adjacent to vertices of $A_\alpha^j$, this
		amounts to at most $h(9)+2h(15)+4=148$.

		We then observe that the number of contractions of pairs of vertices of $\widetilde{A}_1^j \cup
		\widetilde{A}_2^j \cup \widetilde{A}_3^j$ that may happen when obtaining $A^j$ is at most 5
		for the following reasons.
		We have at most two contractions to contract the potential vertices with empty
		neighbourhoods coming from each $F_i$. Furthermore, at most 3 vertices of the $P_i$ can have
		adjacent vertices in two $F_\alpha$ (the first vertices of $F$ on the path from each $v_i$ to $r$
		in $T$), so we may contract the two potential representatives of the neighbourhood classes
		consisting of a singleton of such a vertex in the two adjacent $F_\alpha$. 
		Since we know $|A^j| \leq h(15)$ and each contraction may reduce the number of vertices by
		at most $1$, we have $|\widetilde{A}_1^j| + |\widetilde{A}_2^j| + |\widetilde{A}_3^j| \leq
		h(15) + 5$. 

		The red degree of a vertex of $Q'_i$ is bounded by the sizes of the $|\widetilde{A}_\alpha^j|$
		of its 3 adjacent layers on the two faces to which it is adjacent, this is because by always
		contracting to the same vertex in each neighbourhood class we can ensure that the number of
		red edges to this vertex is always increasing. If we add the size of the last face for each
		layer (positive terms), we can easily bound using the previous inequality, by
		$3(h(15)+5)=183$.

		The red degree of a vertex of layer $j$ when contracting to form $A^j$ is at most
		$|A^{j-1}|+|\widetilde{A}_1^j|+|\widetilde{A}_2^j|+|\widetilde{A}_3^j|+
		|\widetilde{A}_1^{j+1}|+|\widetilde{A}_2^{j+1}|+|\widetilde{A}_3^{j+1}|+6-2$, 
		where the $6$ term bounds the number of vertices of $Q'_i$ in the layers $j$ and $j+1$.
		combining previous inequalities, we may bound by $3h(15)+ 2\cdot 5 + 4 = 182$.
	\end{proof}

	When the outerface is reached, we can contract arbitrarily to a single vertex layer by layer,
	and then contract the path. Doing so we have red degree at most $3h(9)+1<183$ because there are only
	$3$ vertices on the outerface.

	We conclude that we have constructed a $d$-contraction sequence of $G$ such that $d \leq 183$.
\end{proof}

\section{Bipartite graph}

\begin{theorem}
	The twin-width of the universal bipartite graph $\mathcal{B}(n)$ is $n - \log(n) + \mathcal{O}(1)$.
\end{theorem}

\begin{proof}
We first prove an upper bound.
Let $k \in [n]$.
We denote by $A$ a subset of $k$ vertices in $X=[n]$.
First, contract vertices of $Y=2^{[n]}$ that have the same neighbourhood in $A$.
When this is done, vertices of $A$ have no incident red edges, while vertices of $X \setminus A$
have red edges going to all remaining vertices of $Y$ (there are $2^k$ such vertices).

At this point the red degree is at most $\max(2^k,n-k)$.

The vertices of $X \setminus A$ can then be contracted into a single vertex 
without creating new red edges.
We can then contract all the remaining vertices of $Y$ into a new vertex of red degree $k+1$.
Finally, we contract $A$ onto the said vertex. This establishes that for any choice of $k$

\[\tww(\mathcal{B}(n)) \leq \max(2^k,n-k,k+1).\]

By choosing $k=\left\lfloor \log(n) - 1\right\rfloor$, we obtain
$\tww(\mathcal{B}(n)) \leq n - \log(n) + \mathcal{O}(1)$.

We now prove a lower bound.
Consider a $(n-k)$-contraction sequence for $\mathcal{B}(n)$.
We focus on the moment before the first contraction with a vertex of $X$.

Note that the number of initial vertices contained in a current vertex of $Y$ with red degree $d$
is at most $2^d$, hence at most $2^{n-k}$.

Since a contracted vertex of $Y$ has red degree at least $1$. From the bound on the red degree of
vertices of $X$, we know that there are at most $n(n-k)$ red edges. More precisely, if we denote
by $l_a$ for $a \in [n-f(n)]$ the number of vertices of $Y$ with red degree $a$, we have
\[\sum_{a=1}^{n-k} al_a \leq n(n-k).\]
The number of vertices that were contracted in $Y$ is therefore at most 
\[\sum_{a=1}^{n-k} l_a2^a = \sum_{a=1}^{n-k} al_a \cdot \frac{2^a}{a} \leq n(n-k) \cdot \frac{2^{n-k}}{n-k} = n2^{n-k}.\]

When contracting with a vertex of $X$ for the first time,
the number of red edges that become incident to it is therefore at least
\[ 2^{n-1} - n2^{n-k} - 1.\]
This is bounded by $n-k$, which implies 
$k \leq \log_2(n)+\mathcal{O}(1)$.

We can thus conclude that 
\[ \tww(\mathcal{B}(n)) = n - \log_2(n) + \mathcal{O}(1).\]
\end{proof}

\section{Conclusion}

Although, we provide no lower bound matching our upper bound on the twin-width of graphs of bounded
treewidth, we believe that the exponential dependency is necessary. One might want to consider
k-trees with heavy branching in order to find such a lower bound.

As for the twin-width of planar graphs, it might be possible to improve the given bound with a more
careful analysis. Another interesting prospect would be to adapt our arguments for planar graphs to
graphs of bounded genus, for which properties of the embedding might also prove useful.

\bibliography{main}
\bibliographystyle{alpha}

 \end{document}